\documentclass[10pt,conference]{IEEEtran}

\usepackage[utf8]{inputenc}
\usepackage[mathscr]{euscript}
\usepackage{graphicx}
\usepackage{amssymb}
\usepackage{amsmath}
\usepackage{stmaryrd}
\usepackage{xfrac} 
\usepackage{faktor}
\usepackage{float}
\usepackage{diagbox}
\usepackage[normalem]{ulem}

\usepackage{amsthm}
\usepackage{amsmath}
\usepackage{amsfonts}
\usepackage{stmaryrd}
\usepackage{enumitem}

\newtheorem{theorem}{Theorem}

\newtheorem{proposition}{Proposition}
\newtheorem{lemma}{Lemma}
\newtheorem{example}{Example}
\newtheorem{definition}{Definition}
\newtheorem{remark}{Remark}

\usepackage{changepage} 

\usepackage[english]{babel}
\usepackage{relsize}

\usepackage{algpseudocode}

\usepackage{algorithm}

\usepackage{multirow}
 
\usepackage[htt]{hyphenat}
\usepackage{url}

\usepackage{hyperref}
\providecommand{\doi}[1]{\textsc{doi}: \href{http://dx.doi.org/#1}{\nolinkurl{#1}}}

\usepackage{xspace}
\usepackage{color}

\usepackage[export]{adjustbox}
\usepackage{fancybox}
\newcommand{\resultbox}[1]{\noindent\shadowbox{\parbox{.959\columnwidth}{#1}}}

\usepackage{enumitem}

\newenvironment{packed_description}{
\begin{description}[leftmargin=*]
  \setlength{\itemsep}{1pt} 
  \setlength{\parskip}{0pt}
  \setlength{\parsep}{0pt}
}{\end{description}}

\usepackage{listings}
\definecolor{DarkGreen}{rgb}{0,0.3,0}
\definecolor{DarkBlue}{rgb}{0,0,0.7}
\definecolor{grey}{rgb}{0.5,0.5,0.5}
\definecolor{redOrange}{rgb}{0.8,0.2,0}
\definecolor{purple}{rgb}{0.3,0,0.8}

\lstdefinelanguage{Spectra}[]{}{
  commentstyle=\color{grey}\itshape,
  keywordstyle=[1]\color{blue}\bfseries,
  keywordstyle=[2]\color{DarkGreen}\bfseries,
  keywordstyle=[3]\color{red}\bfseries,
  keywordstyle=[4]\color{purple}\bfseries,
  morekeywords=[1]{G,GF,next,H,ONCE,boolean,false,true,spec,alw,always,alwaysEventually,alwEv,trans,ini,initially,iff},
  morekeywords=[2]{out,sys,gar},
  morekeywords=[3]{in,env,asm},
  morekeywords=[4]{counter,overflow,inc,reset,underflow,type,monitor,counter,define,predicate},
  morecomment=[l]{//},
  morecomment=[s]{/**}{*/}
}

\makeatletter
\def\lst@makecaption{%
  \def\@captype{table}%
  \@makecaption
}
\makeatother

\newcommand{\op}[1]{\textbf{\texttt{\color{blue}#1}}}

\usepackage{xspace}

\newcommand{\ddmin}{\texttt{DDMin}\xspace}
\newcommand{\qxplain}{\texttt{QuickXplain}\xspace}

\newcommand{\quickcore}{\texttt{QuickCore}\xspace}
\newcommand{\punch}{\texttt{Punch}\xspace}

\begin{document}

\title{Unrealizable Cores for \\Reactive Systems Specifications}

\author{
\IEEEauthorblockN{Shahar Maoz}
\IEEEauthorblockA{Tel Aviv University\\ Tel Aviv, Israel}
\and
\IEEEauthorblockN{Rafi Shalom}
\IEEEauthorblockA{Tel Aviv University\\ Tel Aviv, Israel}
}

\maketitle

\begin{abstract}

One of the main challenges of reactive synthesis, an automated procedure to
obtain a correct-by-construction reactive system, is to deal with unrealizable
specifications. One means to deal with unrealizability, in the context of GR(1),
an expressive assume-guarantee fragment of LTL that enables efficient synthesis,
is the computation of an unrealizable core, which can be viewed as a
fault-localization approach. Existing solutions, however, are computationally
costly, are limited to computing a single core, and do not correctly support
specifications with constructs beyond pure GR(1) elements.

In this work we address these limitations. First, we present QuickCore, a novel
algorithm that accelerates unrealizable core computations by relying on the
monotonicity of unrealizability, on an incremental computation, and on
additional properties of GR(1) specifications.
Second, we present Punch, a novel algorithm to efficiently compute all
unrealizable cores of a specification. Finally, we present means to correctly
handle specifications that include higher-level constructs beyond pure GR(1)
elements.

We implemented our ideas on top of Spectra, an open-source language and
synthesis environment. Our evaluation over benchmarks from the literature shows
that QuickCore is in most cases faster than previous algorithms, and that its
relative advantage grows with scale. Moreover, we found that most
specifications include more than one core, and that Punch finds all the cores
significantly faster than a competing naive algorithm.

\end{abstract}

\section{Introduction}

Reactive synthesis is an automated procedure to obtain a correct-by-construction
reactive system from its temporal logic specification~\cite{PR89}.
GR(1) is an assume-guarantee fragment of Linear Temporal Logic (LTL) that has an
efficient symbolic synthesis algorithm~\cite{BJP+12}. GR(1) specifications
include assumptions and guarantees that specify what should hold in all initial
states, in all states and transitions (safeties), and infinitely often on every
computation (justices). The expressive power of GR(1) covers many well-known LTL
specification patterns~\cite{DAC99,MaozR15}, and it has been recently applied in
several domains, including robotics (see, e.g., ~\cite{ManiatopoulosSP16}).

One of the main challenges of reactive synthesis in general and of GR(1)
synthesis in particular is to deal with unrealizable
specifications~\cite{AlurMT13,CavezzaA17,CimattiRST08,KonighoferHB13,KuventMR17,MaozRS19}.
To help engineers debug unrealizable specifications, several works have
suggested the computation and use of an \textit{unrealizable core}, a locally
minimal subset of guarantees that is sufficient for unrealizability. Computing
the core may be viewed as a fault-localization approach to unrealizability.

However, existing solutions to computing an unrealizable core suffer from three
main limitations. First, core computation in existing solutions is costly, as it
requires many invocations of realizability checking. Second, existing solutions
are limited to finding a single core and thus provide only partial information
about the realizability problems in the specification at hand. Third, existing
solutions are limited to pure GR(1) specifications, and do not correctly handle
specifications with richer language constructs. All these limit the
applicability of core computations as an effective means to dealing with
unrealizability.

In this work we address these three limitations. Our first contribution is
\quickcore, a novel algorithm that accelerates unrealizable core computations.
The effectiveness and correctness of \quickcore are based on four observations.
First, that unrealizability is monotonic.
Second, that core computations can be incremental. Third, that checking the
realizability of specifications with fewer justices is typically significantly
faster than of those with more justices. And fourth, that minimizing initial
guarantees requires only one fixed-point computation.


Our second contribution is \punch, an algorithm to efficiently compute
all unrealizable cores of a specification. Moreover, \punch computes
the intersection of all cores without having to compute all of them. In
particular, \punch is able to quickly check whether a core that was found is the
only one that may be found.


Finally, our third contribution is the extension of core computations to
correctly handle specifications that include, beyond pure GR(1) elements,
higher-level constructs such as patterns, past LTL operators, monitors, and
counters, as supported, e.g., in Spectra~\cite{MaozR21}. This is important in
order to correctly apply core computations to these more compact and readable
specifications, and to present the results not at the level of the internal
representation but at the abstraction level used by the engineer who wrote the
specification. 




It is important to note that while \quickcore is specific to GR(1)
unrealizability, \punch is in fact a generic algorithm for computing all cores.
Indeed, the definition and correctness of \punch in Alg.~\ref{alg:punch} is
independent of realizability checking or GR(1). Thus, \punch can be used, e.g.,
to compute all vacuity cores~\cite{MaozS20}, all realizable cores (subsets of
assumptions that suffice for realizability), and more generally, all locally
minimal subsets for any monotonic criterion (e.g., unsatisfiability). Similarly,
our third contribution on correctly dealing with higher-level constructs beyond
pure GR(1) applies not only to unrealizable cores, but, in principle, to any
similar analysis of specifications with higher-level constructs that are reduced
to GR(1). We consider this generality a nice advantage of our work, with
possible future applications beyond the focus of the present paper.

We implemented our ideas on top of Spectra, a rich specification language and
open source tool set for reactive synthesis~\cite{MaozR21}. All our algorithms
are implemented on top of, and compared to, recently suggested heuristics for
realizability and core computations for GR(1)~\cite{FirmanMR20}. We validated
and evaluated our work on benchmarks from the literature. The evaluation shows
that \quickcore is in almost all cases faster than previous algorithms and that
many specifications indeed include more than one core. It further shows that
\punch is much faster than a competing naive algorithm, and that it is able
to compute almost all cores for the SYNTECH benchmarks in reasonable times.

Means to deal with unrealizability of temporal specifications have been studied
in the literature. Beyond unrealizable cores~\cite{KonighoferHB13}, these also
include different approaches to counter-strategy
generation~\cite{KuventMR17,MaozS13} and assumption refinement or
repair~\cite{AlurMT13,CavezzaA17,MaozRS19}. We discuss related work in
Sect.~\ref{sec:related}.

\section{Running Example}
\label{sec:example}


\lstset{language=Spectra}
\begin{lstlisting}[label=lst:spec, float=t, caption={\small Lift controller specification, adopted from~\cite{MaozS20}},
belowskip=-2em,aboveskip=3pt]
env boolean b1;
env boolean b2;
env boolean b3;

sys Int(1..3) f;

// No buttons are initially pressed
asm !b1 and !b2 and !b3;

// Request is removed when satisfied
asm G ( (b1 and f=1) -> next(!b1));
asm G ( (b2 and f=2) -> next(!b2));
asm G ( (b3 and f=3) -> next(!b3));

// Request must remain while unsatisfied
asm G ( (b1 and f!=1) -> next(b1));
asm G ( (b2 and f!=2) -> next(b2));
asm G ( (b3 and f!=3) -> next(b3));

// Lift is initially at lowest floor
gar f=1;

// Always stay at same floor or move to adjacent floor
gar G (f>=next(f)-1 and f<=next(f)+1);

// Do not move up when there are no requests
gar G (f<next(f)) ->(b1 or b2 or b3);

// Eventually grant each request
gar GF (b1 -> f=1);
gar GF (b2 -> f=2);
gar GF (b3 -> f=3);

// Visit every floor infinitely often
gar GF f=1;
gar GF f=2;
gar GF f=3;
\end{lstlisting}

As a running example, we use a lift specification (see Lst.~\ref{lst:spec}),
taken from~\cite{MaozS20}, which has appeared in several variants in previous
GR(1)-related papers~\cite{AlurMT13, BJP+12, CavezzaA17, PitermanPS06}. The
specification is written in Spectra format~\cite{MaozR21,SpectraWebsite}.
It is small and simple, to fit the paper presentation. In our evaluation we have
used larger and more complex specifications, taken from benchmarks.

The specification models a controller for a three floors lift. The lift has
three request buttons, one on each floor. Requests are represented by
environment variables \texttt{b1}, \texttt{b2}, and \texttt{b3}, which may be
independently true or false. The current floor of the lift is represented by the
system variable \texttt{f}. The environment is required to initially have no
requests (l.~8), turn off any granted request at the next step (ll.~11-13), and
keep ungranted requests (ll.~16-18). The system is required to start the lift on
the first floor (l.~21), and to disallow the lift to move more than one floor at
a time (l.~24). The system is also required not to move up when there are no
requests (l.~27), to eventually grant every request (ll.~30-32), and to make
sure every floor is visited infinitely often (ll.~35-37).

The specification is unrealizable but it is not easy to see why and debug it.
So, the engineer may want to employ a fault-localization approach and find an
unrealizable core, a locally minimal subset of guarantees that is sufficient for
unrealizability. A modified specification with only these guarantees is already
unrealizable; removing any one guarantee from that specification, will render it
realizable.

By running our new algorithm \quickcore
the engineer finds that the set of guarantees in lines \{21,27,36\} is a core.

However, an unrealizable core is not necessarily unique, i.e.,
the specification may induce additional unrealizable cores.
Our new algorithm \punch, first finds the above core, then finds that the
guarantee in line 27 is the intersection of all the cores, and finally reports
all five remaining cores, namely \{21,27,37\}, \{27,35,36\}, \{27,35,37\},
\{27,36,37\}, and \{24,27,30,37\}. As our evaluation shows, having more
than one core is indeed rather common.

Note that the early detection of the intersection of all the cores indicates
whether additional cores exist. Moreover, the intersection is of interest, since
making the specification realizable by removing or weakening only one guarantee
is only possible with the guarantees in the intersection. Finally, the size of
the intersection serves as a lower bound on the size of any of the cores induced
by the specification.

\section{Preliminaries}



\subsection{LTL, GR(1), and Realizability}
\label{subsec:realize}

We use a standard definition of \emph{linear temporal
logic (LTL)}, e.g., as found in~\cite{BJP+12}, over present-future
temporal operators \op{X} (next), \op{U} (until), \op{F} (finally), and
\op{G} (globally), and past temporal operator \op{H} (historically),
for a finite set of Boolean variables $\mathcal{V}$.
LTL formulas can be used as specifications of reactive systems, where
atomic propositions are interpreted as environment (input) and
system (output) variables. An assignment to all variables is called
a state. 

GR(1) is a fragment of LTL. A GR(1) specification contains initial assumptions and guarantees over initial
states, safety assumptions and guarantees relating the current and next
state, and justice assumptions and guarantees requiring that an
assertion holds infinitely many times during a computation.
We use the following abstract syntax definition of a GR(1) specification taken
from~\cite{MaozS20}.

\begin{definition} [Abstract syntax of a specification]
\label{def:abst}
A GR(1) specification is a tuple $Spec = \langle V_e, V_s, D, M_e, M_s\rangle$,
where $V_e$ and $V_s$ are sets of environment and system variables respectively,
$D:V_e\cup V_s \rightarrow Doms$ assigns a finite domain to each
variable\footnote{The use of any finite domain rather than only Boolean
variables is straightforward and supported by many tools, including Spectra.},
and $M_e$ and $M_s$ are the environment and system modules. A module is a
triplet $M=\langle I, T, J\rangle$ that contains sets of initial assertions
$I=\{I_n\}_{n=1}^{i}$, safety assertions $T=\{T_n\}_{n=1}^{t}$, and justice
assertions $J=\{J_n\}_{n=1}^{j}$ of the module, where $i=|I|, t=|T|$ and
$j=|J|$.
The set of elements of module $M=\langle I, T, J\rangle$ is $B_M = I\cup
\{\op{G}~T_i\}_{i=1}^{t} \cup \{\op{GF}~J_i\}_{i=1}^{j}$.
\end{definition}

Given a set $\mathcal{Z}$ of variables, $\mathcal{Z}' = \{\op{X} v |
v\in \mathcal{Z}\}$ contains a copy of its variables at the next state.
Let $M_e=\langle I_e, T_e, J_e\rangle$, $M_s=\langle I_s, T_s, J_s \rangle$, and
$\mathcal{V} = V_e\cup V_s$.
Then, the elements of $I_e, T_e, J_e, I_s, T_s$ and $J_s$ are propositional
logic expressions over $V_e, \mathcal{V} \cup V'_e, \mathcal{V}, \mathcal{V},
\mathcal{V} \cup \mathcal{V}'$ and $\mathcal{V}$ respectively.

GR(1) has efficient symbolic algorithms for realizability checking and 
controller synthesis, presented in~\cite{BJP+12,PitermanPS06}.
For this a game structure of a two-player game $G=\langle \mathcal{V},\mathcal{X}, \mathcal{Y}, 
\theta_e, \theta_s, \rho_e, \rho_s,\varphi\rangle$ is defined. The GR(1) game has a set of variables
$\mathcal{V}= V_e\cup V_s$, environment and system variables ($\mathcal{X}=V_e$ and
$\mathcal{Y}=V_s$ resp.), environment and system initial states ($\theta_e=\wedge_{d\in I_e}d$ 
and $\theta_s=\wedge_{d\in I_s}d$ resp.), environment and system transitions ($\rho_e = \wedge_{t\in T_e}t$
and $\rho_s = \wedge_{t\in T_s}t$  resp.), and a winning condition 
$\varphi=\bigwedge_{j \in J_e} \op{GF}j \rightarrow \bigwedge_{j \in J_s} \op{GF}j$.

A GR(1) specification is realizable, i.e., allows an implementation, iff
the system wins the game.
Roughly, this means that if the environment keeps all initial
assumptions then the system should keep all initial guarantees, as long as the
environment keeps all safety assumptions the system should keep all safety
guarantees, and in all infinite plays, if the environment keeps all justice
assumptions the system should keep all justice guarantees.

For this the algorithm of~\cite{BJP+12,PitermanPS06} computes a winning region which is a
set of winning states from which the system has a winning strategy. 
A winning strategy prescribes the outputs of
a system for all possible environment choices that allows the system to win.
The winning region is computed according to a fixed-point computation over transitions
and justices alone. 
GR(1) realizability checks if for all initial
environment choices the system can enter a winning state.
GR(1) synthesis computes a winning strategy, if one exists.

\subsection{Monotonic Criteria and Cores}

Given a set $E$, and a monotonic criterion on subsets of $E$, a core is a local minimum 
that satisfies the criterion. Formally:

\begin{definition} [Monotonic criterion]
A Boolean criterion over subsets of $E$ is monotonic iff for any two sets $A,B$ such that 
$A\subseteq B\subseteq E$, if $A$ satisfies the criterion then $B$ satisfies the criterion.
\end{definition}

\begin{definition} [Core] 
Given a set $E$ and a monotonic criterion over its subsets, a set $C\subseteq E$ is a core of
$E$ iff $C$ satisfies the criterion, and all its proper subsets $C'\subset C$ do not
satisfy the criterion.
\end{definition}

Unrealizability is monotonic w.r.t. subsets of guarantees, i.e., 
adding guarantees to an unrealizable specification
keeps it unrealizable. Intuitively, this is so because adding
guarantees strengthens the constraints on the system, and
does not change the constraints on the environment. Formally:

\begin{proposition} [Unrealizability is monotonic]
\label{prop:monotonic}
Given two specifications, $Spec_1 = \langle V_e, V_s, D, M_e, M^1_s\rangle$, and
$Spec_2 = \langle V_e, V_s, D, M_e, M^2_s\rangle$, such that $B_{M^1_s}\subseteq B_{M^2_s}$.
Then, if $Spec_1$ is unrealizable, $Spec_2$ is also unrealizable. 
\end{proposition}

\subsection{Existing Domain-Agnostic Minimization Algorithms}
\label{subsec:ddmin}

We recall three existing domain-agnostic minimization algorithms from the
literature, namely delta debugging~\cite{ZellerH02} (\ddmin),
\qxplain~\cite{Junker04, MarquesJB13}, and linear minimization, which we
denote by \texttt{LinearMin}. All three algorithms find a core of a set $E$,
given a monotonic criterion \texttt{check}.

\begin{algorithm}
\caption{The delta debugging algorithm \ddmin from~\cite{ZellerH02} as a
recursive method that minimizes a set of elements $E$ by partitioning
it into $n$ parts (initial value $n=2$)}
\label{alg:ddmin}
\begin{algorithmic}[1]\footnotesize
  \For {$part \in partition(E, n)$}\label{alg:ddmin:parts}
    \If {\textbf{check}($part$)}\label{alg:ddmin:success1}
    	\State {return \textbf{ddmin}($part$, 2)}\label{alg:ddmin:n2}
    \EndIf
  \EndFor
  \For {$part \in partition(E, n)$}\label{alg:ddmin:complements}
    \If {\textbf{check}($E \setminus part$)}\label{alg:ddmin:success2}
      \State {return \textbf{ddmin}($E \setminus part$, $max(n-1,2)$)}\label{alg:ddmin:decrease}
    \EndIf
  \EndFor
  \If {$n \geq |E|$}
    \State {return$E$} \label{alg:ddmin:terminate}
  \EndIf
  \State {return \textbf{ddmin}($E$, $min(|E|, 2n)$)}\label{alg:ddmin:increase}
\end{algorithmic}
\end{algorithm}

We show a pseudo-code for \ddmin in Alg.~\ref{alg:ddmin}.
The inputs for the algorithm are a set $E$ and the number $n$ of parts of $E$ to
check. The algorithm starts with $n=2$ and refines $E$ and $n$ in recursive
calls according to different cases
(ll.~\ref{alg:ddmin:n2},~\ref{alg:ddmin:decrease},
and~\ref{alg:ddmin:increase}).
The computation starts by partitioning $E$ into $n$ subsets and evaluating
\texttt{check} on each subset $part$ and its
complement.
If \texttt{check} holds (l.~\ref{alg:ddmin:success1} or
l.~\ref{alg:ddmin:success2}), the search is continued recursively on the
subset $part$ or on its complement respectively. If \texttt{check} holds neither on any
subset $part$ nor on the complements, the algorithm increases the granularity of
the partitioning to $2n$ (l.~\ref{alg:ddmin:increase}) and recurs, or
terminates when the granularity is not smaller than the size of $E$
(l.~\ref{alg:ddmin:terminate}).
\ddmin has quadratic worst-case complexity and logarithmic best-case complexity
in terms of $|E|$.


\qxplain is a recursive divide and conquer algorithm that minimizes each
half, one after the other, in an incremental way (see
Sect.~\ref{subsec:minwbase}). It has a worst-case complexity of $O(k+klog(
{|E|\over k} ))$, where $k$ is the size of the largest core. To the best of our
knowledge, \qxplain was never previously applied to computing unrealizable
cores for reactive systems specifications.

Finally, \texttt{LinearMin} was originally suggested in~\cite{CimattiRST08}, and
compared with \ddmin in~\cite{KonighoferHB13}.
\texttt{LinearMin} goes over elements of the input set one by one, and removes
an element iff the criterion holds for the set without the element.
An example of \texttt{LinearMin} can be found in
ll.~\ref{alg:qc:begini}-\ref{alg:qc:endini} of Alg.~\ref{alg:qc}. The complexity
of \texttt{LinearMin} is linear in $|E|$.

\subsection{Minimization With a Base and Incremental Core Computation}
\label{subsec:minwbase}
We define a notion of minimization with a base as follows.

\begin{definition}[Minimization with a base]
\label{def:minwbase}
Assume a set $E$ of elements, two disjoint subsets $Base,A\subseteq E$,
and a minimization algorithm \texttt{Alg} that detects cores
according to a monotonic criterion \texttt{check}. Assume also that the set $Base\cup A$
satisfies the criterion.
We denote by $MinWBase(Alg, E, Base, A, check)$ an algorithm that computes a locally minimal
$A'\subseteq A$ s.t. $Base\cup A'$ satisfies the criterion, by applying \texttt{Alg} to $A$,
and replacing every \texttt{check}$(X)$ operation with \texttt{check}$(Base\cup X)$.
\end{definition}

Monotonicity is ensured for minimization with a base
because $A_1\subseteq A_2$ implies $Base\cup A_1\subseteq Base\cup A_2$.

Note that $Base$ does not have to be a subset of a core of $E$, and it may contain a core.
$Base$ only has to satisfy $A\cap Base=\emptyset$, and that $Base\cup A$ satisfies the
criterion. The following holds trivially.

\begin{proposition}
\label{prop:emptymin}
Given the notations of Definition~\ref{def:minwbase},
\begin{enumerate}[label=(\roman*)]
\item \label{prop:emptymin:coreoutcome} If $Base$ is a subset of all cores of $E$, then $Base\cup A'$ is a core of $E$. 
\item \label{prop:emptymin:emptyoutcome} If $Base$ contains a core of $E$, then $A'=\emptyset$.
\end{enumerate}
\end {proposition}
 
Finally, we use the idea of incremental minimization in some of our algorithms. 
Lemma~\ref{lem:parts} states that incrementally minimizing two subsets that partition a set 
produces a core.

\begin{lemma} [Incremental Core Computation]
\label{lem:parts}
Let $A$ and $B$ be disjoint sets such that $E=A\cup B$ satisfies a monotonic criterion.
Let $A'$ be a locally minimal subset of $A$ such that $A'\cup B$ satisfies the criterion,
and $B'$ be a locally minimal subset of $B$ such that $A'\cup B'$ satisfies the criterion.
Then $A'\cup B'$ is a core of $E$.
\end{lemma}

\begin{proof}
By definition $A'\cup B'$ satisfies the criterion. $A'$ exists because $E$ satisfies the
criterion, and $B'$ exists because $A'\cup B$ satisfies the criterion.

To prove that $A'\cup B'\subseteq E$ is locally minimal, let $x\in A'\cup B'$. If $x\in A'$
then by definition of $A'$, $(A'\setminus \{ x \}) \cup B$ does not satisfy the criterion, thus by
monotonicity $(A'\setminus \{ x \}) \cup B'$ does not satisfy the criterion either. 
Otherwise $x\in B'$, and by definition
of $B'$, $A'\cup (B'\setminus \{ x \})$ does not satisfy the criterion.

\end{proof}

\section{\quickcore}
\label{sec:quickcore}

We are now ready to present our first contribution, the \quickcore algorithm,
which aims to accelerate unrealizable core computations. The correctness
and efficiency of \quickcore rely on the following observations.
First, that unrealizability is monotonic (see Prop.~\ref{prop:monotonic}).
Second, that it is possible to compute a core incrementally (see
Lemma~\ref{lem:parts}). Third, that checking the realizability of specifications
with fewer justices is typically much faster. Fourth, that removing justice
assumptions from an unrealizable specification that has no justice guarantees
preserves unrealizability and preserves winning regions. And finally, that
minimizing the initial assertions requires only a single computation of the
winning region of the system, plus a small constant number of symbolic operations.

\subsection{\quickcore Algorithm Overview}

Roughly, \quickcore begins by trying to remove as many justices as
possible. Therefore, if no justice guarantees are required for a core it removes all of them, and
all justice assumptions. Otherwise, it minimizes justice guarantees alone using \ddmin with a base.
Later, \quickcore minimizes safety guarantees using \ddmin with a base.
Finally, \quickcore minimizes all initial guarantees using \texttt{LinearMin}.


Algorithm~\ref{alg:qc} presents \quickcore in pseudo-code. 
\quickcore minimizes guarantees,
group by group, according to their type; first justices, then safeties, and ending with initial assertions.
It begins with a realizability check of the specification without justice guarantees in line~\ref{alg:qc:nojustcheck}. 
If the specification without these is realizable, at least one such justice guarantee is required for unrealizability, 
and so it minimizes justice guarantees alone (i.e., keeping all initial assertions and safeties)
using \ddmin with a base in line~\ref{alg:qc:minjust}.
Otherwise, it determines that the core it computes has no justice guarantees (line~\ref{alg:qc:nojust}) 
and removes all environment justices in line~\ref{alg:qc:noenvjust}
before it continues to look for system safeties and initial assertions.

With a specification that has a minimized set of justice guarantees and maybe no environment justices,
\quickcore uses \ddmin with a base in order to minimize the set of safeties alone
(line~\ref{alg:qc:minsafe}). It keeps initial assertions and the minimized justice guarantees unchanged.

At the last stage of \quickcore, it computes the winning region of the system in
line~\ref{alg:qc:win} for the specification with the minimized set of justices
and safeties. It now uses linear minimization in the loop in
lines~\ref{alg:qc:begini}-\ref{alg:qc:endini}, going over initial assertions one
by one, and checking if the system wins without each. If so, we keep this
initial assertion because we want to maintain unrealizability.
Checking this is done with $SysWin(\theta_e, \theta_s, w)$, which determines if
for every possible environment choice satisfying $\theta_e$, there is a choice
for the system from $\theta_s$ that is inside the winning region $w$.
This is computed with a small, constant number of symbolic operations.


\begin{remark}
Realizability checks in our implementation use the performance heuristics
suggested in~\cite{FirmanMR20}.
\texttt{QuickCore} must disable two of these heuristics, namely, fixed-point recycling and early
detection of unrealizability, before line~\ref{alg:qc:win}, in order to avoid an incomplete winning
region computation.
For example, the latter heuristics may over-approximate the winning region by halting the outer greatest fixed-point
when the winning region becomes small enough to ensure that the system loses.
This is good enough for realizability checks, so we
use it safely before line~\ref{alg:qc:win}, yet using it from that line on may wrongly
keep some initial guarantees.
\end{remark}

\begin{example}
For the specification in Lst.~\ref{lst:spec} (see Sect.~\ref{sec:example}),
\quickcore finds that at least one justice is required for a core in line~\ref{alg:qc:nojustcheck},
and detects the justice in line 36 of the specification with the minimization in line~\ref{alg:qc:minjust}.
Given this justice, the safety in line 27 of the specification is detected in line~\ref{alg:qc:minsafe}.
Lines~\ref{alg:qc:begini}-\ref{alg:qc:endini} decide that given the above two guarantees, the only
initial guarantee of the specification (line 21)  is required for the core. The resulting core is thus
the set of guarantees in lines \{21,27,36\}.
\end{example}

\begin{algorithm}
\caption{\small \textbf{QuickCore} Given an unrealizable specification find a locally minimal subset
of guarantees that keeps it unrealizable }

\label{alg:qc}
\begin{algorithmic}[1]\footnotesize
\Require{An unrealizable specification $S = \langle V_e, V_s, D, M_e, M_s\rangle$, where $M_e=\langle I_e, T_e, J_e\rangle$ and $M_s=\langle I_s, T_s, J_s\rangle$}
\Ensure{An unrealizability core of $S$}

\algorithmiccomment {Begin by minimizing only justice guarantees~~~~~~~~~~~~~~~~~~~~~~~~~}
  \If {\textbf{Realizable}$(\langle V_e, V_s, D, M_e, \langle I_s, T_s, \emptyset\rangle\rangle)$} \label{alg:qc:nojustcheck}
    \State {$J_c\leftarrow \textbf{MinWBase}(\ddmin, B_{M_s}, I_s\cup T_s, J_s,\neg\textbf{Realizable}(S))$}\label{alg:qc:minjust}
  \Else
    \State {$J_c\leftarrow \emptyset$} \label{alg:qc:nojust}
    \State {$M_e\leftarrow\langle I_e, T_e, \emptyset\rangle$} \label{alg:qc:noenvjust}
  \EndIf

\algorithmiccomment {Continue with minimizing only safety guarantees~~~~~~~~~~~~~~~~~~~~}
  \State {$T_c\leftarrow \textbf{MinWBase}(\ddmin, B_{M_s}, I_s\cup J_c, T_s,\neg\textbf{Realizable}(S))$}\label{alg:qc:minsafe}

\algorithmiccomment {End with minimizing only initial guarantees~~~~~~~~~~~~~~~~~~~~~~~~~}
  \State {$w\leftarrow \textbf{ComputeWinRegion}(\langle V_e, V_s, D, M_e, \langle I_s, T_c, J_c\rangle\rangle)$} \label{alg:qc:win}
  \State {$I_c\leftarrow I_s$} 
  \State {$envIni \leftarrow \wedge_{d\in I_e}d$}
  \For {$i \in I_c$} \label{alg:qc:begini} \algorithmiccomment{Use linear minimization for initial guarantees~~~~~~}
  	\State{$I_c\leftarrow I_c\setminus \{i\}$}
    \If {\textbf{SysWin}($envIni, \wedge_{d\in I_c}d, w$)}
  	  \State{$I_c\leftarrow I_c\cup \{i\}$}
    \EndIf
  \EndFor \label{alg:qc:endini}
  
  \State {return $I_c\cup T_c\cup J_c$}
\end{algorithmic}
\end{algorithm}

\subsection{Correctness of \quickcore}
\label{subsec:qccorrect}

\quickcore is correct in the sense that it detects an unrealizable core of the guarantees.
To prove correctness we will show that (1) Finding cores of subsets of guarantees one by one
yields a valid core; (2) Removing all justice assumptions from a specification
with no justice guarantees impact neither the realizability of any subset of guarantees, nor 
the winning region of the system; and (3) It is enough to compute the winning region of the system
once when minimizing initial guarantees.



(1) A core can be computed in a compositional manner, by minimizing pairwise disjoint subsets one by
one, using an extension by induction of Lemma~\ref{lem:parts} to any number of finite pairwise disjoint sets.
Specifically, \quickcore partitions the set of guarantees into
three sets, and minimizes each set with established algorithms for local minimum, namely,
\ddmin and \texttt{LinearMin}.


(2) Removing all justice assumptions when no justice guarantees are needed for a core does not
affect the overall correctness of \quickcore, because realizability checks and computation
of winning regions of the system are unchanged. 

Intuitively, in an unrealizable specification with no justice guarantees, the
environment can (and must) win with \textit{finite} plays of the game. Thus,
even though generally removing assumptions may turn a realizable specification
into an unrealizable one, this does not happen when removing justice assumptions
from a specification with no justice guarantees. For such specifications, when
the system does not lose finitely, it wins the infinite games, regardless of any
justice assumptions. Accordingly, all environment justices in this case are
unhelpful in the sense defined in~\cite{CimattiRST08}, i.e., assumptions whose
removal does not change the realizability induced by any subset of the
guarantees. Note that having several assumptions, each unhelpful alone, does not
mean that they are unhelpful as a set, yet justice assumptions in this case are
unhelpful as a set.



Formally, the game
structure of the GR(1) game~\cite{BJP+12} for both specifications is the same, because in GR(1) games, 
justices appear only in the winning condition 
$\bigwedge_{j \in J_e} \op{GF}j \rightarrow \bigwedge_{j \in J_s} \op{GF}j$. 
When there are no justice guarantees, this condition is 
$\bigwedge_{j \in J_e} \op{GF}j \rightarrow \top$, which is an LTL tautology regardless of the justice
assumptions. This ensures both the equirealizability, which is needed for the correctness of
the minimization of safeties in line~\ref{alg:qc:minsafe}, and the correctness of the 
winning region computation in line~\ref{alg:qc:win}.


(3) Finally, the correctness of the minimization of initial guarantees follows from the fact that checking
realizability has two parts. We first compute the winning region of the system in the GR(1) game,
and then check if the system can reach it given all possible environment initial choices. Since the
winning region depends only on safeties and justices, and we keep them unchanged at this phase of \quickcore,
it is enough to compute the winning region only once.


\subsection{Complexity of \quickcore}

Consider a specification with $n$ guarantees. 
The worst-case complexity of \ddmin is $O(n^2)$ realizability checks~\cite{ZellerH02}.
The same holds for \quickcore.
 
%

\section{\punch}
\label{sec:punch}

We now present our second contribution, the \punch algorithm for computing all cores. 

\punch finds all cores of a set $E$ according to a monotonic criterion.
The algorithm is generic in the sense that it requires a \texttt{check} that
evaluates the monotonic criterion, and a \texttt{computeCore}$(E',B)$ that
provides a core in $E'\subseteq E$, given $E'$ that satisfies the criterion, and
a set $B\subseteq E'$ that is a subset of all the cores in $E'$. In particular,
\punch detects all minimal size cores.
As a byproduct, \punch provides an efficient way to find the intersection of all
the cores, without having to compute all of them. This intersection provides an
early estimate of the size of the smallest cores, and, in particular, an early
verdict on the existence of additional cores.

\begin{algorithm}
\caption{\small \textbf{Punch} Find all cores 
according to a monotonic \textbf{check}}
\label{alg:punch}
\begin{algorithmic}[1]\footnotesize
\Require{A set $E$ of elements such that \textbf{check}(E)=$\top$}
\Require{A set $K\subseteq E$ that is a subset of all cores in $E$}
\Ensure{All cores in $E$}

	\State {$C_0 \leftarrow$ \textbf{computeCore}($E, K$)} \label{alg:punch:core}
	\State {$AllCores\leftarrow \{C_0\}$}
	\State {$CI \leftarrow \emptyset$}
	\State {$Cont \leftarrow \emptyset$}
	\For {$x\in C_0\setminus K$} \label{alg:punch:startsplit}
		\If {\textbf{check}($E\setminus\{x\}$)}   \label{alg:punch:check}
			\State {add $x$ to $Cont$}
		\Else
			\State {add $x$ to $CI$}		\label{alg:punch:ci}
		\EndIf						
	\EndFor 						\label{alg:punch:endsplit}
	
	\For {$x\in Cont$} 				\label{alg:punch:startcollect}
		\State {$AllCores \leftarrow AllCores \cup \textbf{Punch}(E\setminus\{x\}, K\cup CI)$} \label{alg:punch:recur}
	\EndFor							\label{alg:punch:endcollect}
   \State {return $AllCores$}		\label{alg:punch:ret}
 \end{algorithmic}
\end{algorithm}

Algorithm~\ref{alg:punch} presents \punch in pseudo-code, as a recursive algorithm
that takes two parameters as input, a set $E$ that has at least one core, and a set 
$K$ that is a subset of all the cores in $E$.
The recursive algorithm finds all the cores in $E$.
Thus, \punch$(E,\emptyset)$ returns all the cores in $E$.

The algorithm finds its first core $C_0\subseteq E$ by applying \texttt{computeCore} 
in line~\ref{alg:punch:core}. Later, in lines~\ref{alg:punch:startsplit}-\ref{alg:punch:endsplit},
it splits all the elements $x\in C_0\setminus K$
into two sets according to whether $E\setminus\{x\}$ satisfies the criterion. $Cont$ gets
all positives and $CI$ gets all negatives. Finally, in lines~\ref{alg:punch:startcollect}-\ref{alg:punch:endcollect},
for all elements $x \in Cont$,  it considers the punched sets
$E\setminus\{x\}$ (hence the name \texttt{Punch}), and recursively looks for all cores inside them, while adding
$CI$ to $K$. By collecting all such cores into $AllCores$ the computation ends.

\begin{example}

Applying \punch to the example in Lst.~\ref{lst:spec} finds the six cores
of the run described in Sect.~\ref{sec:example}. Specifically, that run
is of \texttt{PQC}($B_{M_s},\emptyset$) (see Sect.~\ref{subsec:instances}).

Two cores consist of the guarantees in lines 21, 27, and one of the guarantees
in lines 36 and 37. Indeed, if the lift starts at the first floor, and moves up
only when there are requests, it may never be able to reach the other two
floors. Three other cores consist of the guarantee in line 27, and pairs of the
guarantees in lines 35-37, which require the elevator to visit two different
floors infinitely. Indeed, the system can be forced to visit the lower of the
two floors, and not to go up (line 27). Finally, the guarantees in lines 24, 27,
30, and 37 allow the environment to force the lift to the first floor (line 30),
the system then may or may not move up to the second floor (line 24), which
allows the environment to keep the lift bellow the third floor (line 27), and
thus fail the guarantee to visit that floor (line 37).
\end{example}

\subsection{Correctness of \punch}

First we prove the following Lemma.


\begin{lemma}
\label{lem:ci}
In running \punch$(E,\emptyset)$, after the loop in lines ~\ref{alg:punch:startsplit}-\ref{alg:punch:endsplit}, 
the set $CI$ contains the intersection of all cores in $E$.  
\end{lemma}

\begin{proof}
$K=\emptyset$, thus $CI=\{x\in C_0 | check(E\setminus\{x\})=\bot\}$. Moreover, 
$check(E\setminus\{x\})=\top$ iff $E\setminus\{x\}$ contains a core iff 
$x$ does not belong to all the cores in $E$.
\end{proof}

We now show that the preconditions hold in all recursive calls, and that
\punch is sound and complete. 

\subsubsection*{The Preconditions of \punch are Met in All Recursive Calls}

\punch requires that $E$ satisfies the criterion, and that $K$ is a subset of all the cores in $E$. 
The two preconditions must hold for the recursive calls in line~\ref{alg:punch:recur}.

According to line~\ref{alg:punch:check}, $\forall x\in Cont$ \texttt{check}$(E\setminus \{x\})=\top$,
which satisfies the first precondition.
According to lines~\ref{alg:punch:startsplit} and~\ref{alg:punch:ci}, for all $x\in Cont$,
$K\cup CI\subseteq E\setminus \{x\}$. 
Now, by assumption, $K$ is a subset of all the cores in $E$, and 
by similar reasoning to that of Lemma~\ref{lem:ci}, $CI$ 
is also a subset of all the cores in $E$. In particular, for all $x\in Cont$, $K\cup CI$ must be a 
subset of all the cores in $E\setminus \{x\}$. 
This satisfies the second precondition.


\subsubsection*{\punch is Sound}

A set is added to our list of cores when it is detected by \texttt{computeCore} in recursive
calls at line 1 of \texttt{Punch}. The preconditions of 
\texttt{computeCore} are met because they match the preconditions of \texttt{Punch}.

\subsubsection*{\punch is Complete}

This follows from Thm.\ref{thm:complete} below.


\begin{theorem}
\label{thm:complete}
Let $C$ be a core of a set $E$. In running \punch~$(E,K)$ 
such that $K$ is a subset of all the cores in $E$, we will have $C\in AllCores$ in 
line~\ref{alg:punch:ret} of Alg.~\ref{alg:punch}.
\end{theorem}

\begin{proof}
By induction on $n=|E|$, notice that for $n=1$ there must be exactly one core of size 1, and the 
algorithm is correct for both possible choices of $K$.

Assume by induction that the claim is correct for all sets strictly smaller than
$n$, and fix a set $|E|=n$ with $C\subseteq E$ a core of $E$. If $C$ is detected
in line 1, $C=C_0$ and we are done. Otherwise, $C\neq C_0$. Since $C$ is a core,
and $C_0$ satisfies the criterion, we know that $C_0\not\subseteq
C$, thus there is an $x\in C_0$ such that $x\notin C$. Moreover, $x\notin K$
because $K\subseteq C$, so $x\in C_0\setminus K$. Since $C$ is a core of $E$ and
$x\notin C$, $C$ is a core of $E\setminus \{x\}$. 
This means (1) that $E\setminus \{x\}$ satisfies the
criterion, which in turn means that $x\in Cont$; and (2) since $|E\setminus
\{x\}|<n$, by induction $C$ is found as a core by the call to \texttt{Punch} in
line~\ref{alg:punch:recur}, when the $x$ of the loop coincides with the $x$ in
the proof. This completes the proof.

\end{proof} 

\subsection{Complexity of \punch}
\label{subsec:punchcomplexity}

In general, the number of cores of a set according to a monotonic criterion may
be exponential in the size of the set~\cite{BendikC20}. Thus, the worst-case
complexity of \punch and of any other algorithm that would enumerate all cores, is
exponential.



Lemma~\ref{lem:ci} shows that obtaining the intersection of all the cores
requires only one core computation plus a number of realizability checks the
size of the core that we found. 

\subsection{Employing \punch to Unrealizable Cores}
\label{subsec:instances}

Employing \punch to compute all unrealizable cores requires an implementation of
\texttt{check} to check for unrealizability. We created two implementations of \punch,
which we label \texttt{PUD} and \texttt{PQC}. In \texttt{PUD}, we implemented
\texttt{computeCore} using \ddmin with a base. Specifically, $computeCore(E, K)$
in line~\ref{alg:punch:core} of \punch is implemented with $K\cup
MinWBase(\ddmin, E, K, E\setminus K, check)$. In \texttt{PQC}, we implemented
\texttt{computeCore} with an extended version of \texttt{QuickCore} that
supports minimization with a base. The version of \texttt{QuickCore} minimizes
justices, safeties, and initial assertions that do not belong to a given base
set. Prop.~\ref{prop:emptymin}\ref{prop:emptymin:coreoutcome} ensures the
correctness of both implementations of \texttt{computeCore}.

\section{Memoization}
\label{sec:memo}

We implemented memoization for our algorithms, which allows us to avoid checking
the criterion whenever the check is redundant, based on the results of prior
checks of the criterion and on the criterion's monotonicity.

Basically, we keep a set of prior positive checks and a set of prior negative
checks of the criterion. Whenever a check for a set is about to occur, if we
already have a positive criterion subset, then we know the set is positive and
we avoid actually checking it. Similarly, if we already have a negative superset
then we know the set is negative. Only if memoization fails we perform an actual
check of the criterion and store the result as positive or negative accordingly.

Two additional features accelerate the required subset checks. First, we keep
each memoized set sorted, which enables linear time subset checks. Second, we
keep the collections of positive and negative results sorted according to the
size of the memoized sets, and look for subsets and supersets of relevant size
only (this is correct because a larger set cannot be a subset of a smaller
set, i.e., if $|B|>|A|$ then $B\not\subseteq A$).

\subsection{Memoization in \quickcore}

Realizability checks are the most computationally expensive parts of \quickcore.
Almost all these checks occur as a part of \ddmin runs within it. We use the
\ddmin implementation in Spectra. In this generic implementation of \ddmin, sets
for which the criterion failed are recorded, and we avoid checking them and
their subsets because monotonicity ensures that they must fail as well. This
heuristics was already mentioned in~\cite{ZellerH02}, and implemented
in~\cite{KonighoferHB13} and in~\cite{FirmanMR20} for unrealizable cores.
On top of it, we use the memoization mechanism we described at the beginning of this section.



%
%

\subsection{Memoization in \punch}\label{subsec:mempunch}

We seek to avoid as many as possible calls to \texttt{check} and to
\texttt{computeCore}.

For \texttt{check}, we use the memoization mechanism we described in the
beginning of this section, and add all cores found by
\texttt{computeCore} to the set of positive sets. In \texttt{PQC} and \texttt{PUD} 
(see Sect.~\ref{subsec:instances}) all \texttt{check} operations share the same
memoization, whether the ones invoked at line~\ref{alg:punch:check} of \punch or
at their particular implementations of \texttt{computeCore}.

For \texttt{computeCore}, when we look for a core of a set $E$ (see
line~\ref{alg:punch:core} of \punch), we use the first core that is a subset of
$E$ in previously found cores, if one exists.
This is important because, for example, having two disjoint cores means that
without memoization, we would run a core computation to unnecessarily seek the
second core for every element of the first core. Memoization ensures that the
number of times we run an actual core computation in \punch is equal to the
number of cores in $E$.



\section{Beyond Pure GR(1) Specifications}
\label{sec:extensions}

We now present our third contribution, correct and efficient core computations
for specifications that are reducible to GR(1), yet include language constructs
beyond pure GR(1), such as patterns, monitors, and past LTL formulas.

\subsection{Reducing Higher-Level Constructs into Pure GR(1)}

Recall that many higher-level language constructs can be reduced to pure GR(1)
form by replacing them with additional auxiliary variables as well as new
guarantees or assumptions. See, e.g.,~\cite{MaozR15,MaozR21}.

A pattern (e.g., the response pattern $\op{G}~ (p\rightarrow \op{F}~q)$, which
is not in pure GR(1) form) is reduced according to a deterministic Buchi
automaton that represents it. The states of the automaton are encoded using new
auxiliary variables, its initial state is encoded using an auxiliary assertion
about the initial values of the auxiliary variables, its transitions are encoded
using an auxiliary safety, and its acceptance condition is encoded into a
justice (assumption or guarantee).


Monitors and counters are constructs that track a certain
value. They are reduced by adding auxiliary variables that encode that value,
and optional auxiliary elements that are assertions about its initial value, and
its current and next values (safeties). For example, the monitor in
Lst.~\ref{lst:mon} adds one Boolean auxiliary variable \texttt{a}, one
auxiliary initial assertion (line 4), and one auxiliary safety (line 5).

The reduction is completed by considering auxiliary variables and elements as
system variables and guarantees respectively. This allows one to apply GR(1)
realizability checks and synthesis.


\subsection{A Simple but Incorrect Approach}

One may suggest that core computations would minimize the GR(1) system module
elements, and then trace back to the elements that induced them in the original
specification. This approach, however, is incorrect. As an example, for the
unrealizable specification in Lst.~\ref{lst:mon}, the (incorrect) core computed
by this approach contains only lines 8 and 9. The reason is that only system
elements are minimized. The incorrect computation ignores the auxiliary monitor
initial assertion at line 4 although without this assertion, the specification
is realizable! If unrealizability is a result of auxiliary elements alone, we
may even incorrectly get an empty core (see Prop.~\ref{prop:emptymin}\ref{prop:emptymin:emptyoutcome}). This
means that we must also consider auxiliary elements for minimization, and at the
same time avoid redundant ones that unnecessarily complicate the reduced
specification and inflate the state space.


\lstset{language=Spectra}
\begin{lstlisting}[label=lst:mon, float=t, caption={An unrealizable specification with a monitor},
belowskip=-2em,aboveskip=3pt]
sys boolean b;

monitor boolean a {
	!a; // initially false
	G a->next(a); // once true, remain true forever
}

gar b;
gar G b iff a;
\end{lstlisting}

\subsection{Our Approach}

We have implemented a framework to correctly handle specifications that include
high-level constructs. The framework relies on two-way traceability between the
high-level language construct and the GR(1) elements it reduces to.

Specifically, each distinguishable specification construct, as written by the
engineer, is assigned an ID that represents all of the GR(1) elements it reduces
to. Thus, all elements induced from patterns and past LTL operators are assigned
the ID of the high-level element they belong to, while each assertion inside
monitors and counters has its own ID.

Our implementation builds system and environment modules according to subsets of
IDs. Core computation begins with a set of all environment IDs for the
assumptions, and a set of auxiliary and system IDs for the guarantees. It
performs realizability checks given subsets of IDs. For example, \quickcore may
eliminate all justice assumptions (line~\ref{alg:qc:noenvjust}). If any of these
justice assumptions were induced by patterns, the produced environment module
avoids not only these justices but also their matching pattern-induced initial
and safety assertions, as well as the pattern-induced auxiliary variables that
encoded the states of this pattern's automaton. These are the exact sets of
elements and variables that match the subset of high-level assumptions.

In Lst.~\ref{lst:mon}, the correct core we
compute in this way includes lines 4, 8, and 9.  Together, they are sufficient
for unrealizability, and each of them is necessary.

\section{Evaluation}
\label{sec:evaluation}

We have implemented our ideas on top of Spectra~\cite{MaozR21,SpectraWebsite},
with the performance heuristics from~\cite{FirmanMR20}. Our implementation
includes \quickcore and \punch. For the purpose of evaluation, it also includes
an instance of the \ddmin algorithm implemented in Spectra, an implementation of
\qxplain~\cite{Junker04, MarquesJB13}, and a naive top down algorithm for
computing all cores we label \texttt{TD} (see below). All the above
implementations take advantage of the memoization described in
Sect.~\ref{sec:memo}.

Means to run our implementation, all specifications used in our evaluation, and
all data we report on below, are available in supporting materials for
inspection and reproduction~\cite{unrealcoreWebsite}.
We encourage the interested reader to try them out.

The following research questions guide our evaluation.  

\begin{packed_description}
\item[R0] Which existing domain-agnostic minimization algorithm is the most efficient in our setup?
\item[R1] Is \quickcore efficient, in terms of the number of realizability checks and running time, in comparison to previous algorithms?
\item[R2] Is \quickcore effective, in terms of the size of the core it finds, in comparison to previous algorithms?
\item[R3] Are specifications with multiple unrealizable cores common and how many such cores do most specifications have?
\item[R4] Is \punch efficient in detecting all cores?
\end{packed_description}

%
%

Below we report on the experiments we have conducted in order to answer the above questions. 



\subsection{Corpus of Specifications}
\label{subsec:material}

We use the benchmarks \texttt{SYNTECH15} and
\texttt{SYNTECH17}~\cite{FirmanMR20,MaozR21,SpectraWebsite}, which include a
total of 227 specifications of 10 autonomous Lego robots, written by 3rd year
undergraduate computer science students in a project class taught by the authors
of~\cite{FirmanMR20}. We use all the unrealizable GR(1) specifications from
these, i.e., 14 unrealizable specifications from \texttt{SYNTECH15} (which we
label \texttt{SYN15U}) and 26 unrealizable specifications from
\texttt{SYNTECH17} (which we label \texttt{SYN17U}).


In addition, we used 5 different sizes of \texttt{AMBA}~\cite{BloemGJPPW07} and
of \texttt{GENBUF}~\cite{BloemGJPPW07entcs} (1 to 5 masters, 5
to 40 senders resp.), from each of the 3 variants of unrealizability described
in~\cite{CimattiRST08}. We label these 30 specifications by \texttt{AM+GN}. Note
that these specifications are structurally synthetic and artificially inflated.
We therefore report on their performance (in R0, R1, and R4) but not on their more
qualitative aspects (R2 and R3). Still, the supporting
materials~\cite{unrealcoreWebsite} include all the data we collected.

\subsection{Validation}
\label{subsec:validation}

We have implemented an automatic test that checks that every core that we found
is indeed a locally minimal subset of the guarantees that maintains
unrealizability. We run this check over logs of cores produced by our
algorithms, independent of their original detection. We also verified that the
different algorithms that compute all the cores of a specification (i.e., ones
that terminated before the timeout was reached) found the same number of cores.


\subsection{Experiments Setup}

We ran all experiments on an ordinary PC, Intel Xeon W-2133 CPU 3.6GHz, 32GB RAM
with Windows 10 64-bit OS, Java 8 64Bit, and CUDD 3 compiled for 64Bit, using
only a single core of the CPU. 

Times we report are average values of 10 runs, measured by Java in milliseconds.
Although the algorithms we deal with are deterministic, we performed 10 runs
since JVM garbage collection and BDD dynamic-reordering add variance to running
times.


We used a timeout of 10 minutes for the algorithms that compute all cores, and no timeout
for the algorithms that find a single core.

\subsection{Results: Existing Domain-Agnostic Algorithms in Our Setup}
\label{subsec:qxvsddmin}

\begin{table}\small
\caption{\small Efficiency of \ddmin vs. \qxplain}\vspace{-.5em}
\label{tbl:ddminvsqx}
\resizebox{\columnwidth}{!}{
\begin{tabular}{l || r | r | r | r | r}
Spec set & $\leq 0.1s$ &  0.1-1s & 1-10s & 10-100s & $\geq 100s$ \\
\hline
SYN15U & 0.10 & 0.48 & - &-&-\\
SYN17U & 0.08 & 0.28 & 0.41 & 0.52 & 0.84 \\
\hline
AM+GN  & 0.10 & 0.40 & 0.42 & 0.80 & 0.99 \\

\end{tabular}
}
\end{table}

In Sect.~\ref{subsec:ddmin} we discussed three existing domain-agnostic algorithms
for finding a local minimum given a monotonic criterion. We also noted that
\ddmin was compared and found superior to \texttt{LinearMin} in~\cite{KonighoferHB13}.

Table~\ref{tbl:ddminvsqx} presents the performance of \ddmin versus \qxplain.
The columns show the geometric mean of the ratio of the running times (namely,
the running time of \ddmin divided by that of \qxplain), dissected by the
running time range obtained for \ddmin. We use \texttt{`-'} to mark cases
in which no specifications had \ddmin running time within the corresponding
range. For example, the number 0.28 in the second row means that for
\texttt{SYN17U} specifications for which a core was found by \ddmin in
between 0.1 and 1 seconds, the geometric mean indicates that \ddmin was more
than three times faster than \qxplain.

The results show that \ddmin is more efficient than \qxplain on all our
corpus (although the gap lessens with scale). This justifies our choice of
\ddmin as the domain-agnostic minimization algorithm inside \quickcore and
\texttt{PUD}. It also justifies our choice to use \ddmin as the baseline
algorithm for examining the efficiency of \quickcore (see R1 below).

\vspace{.5em}
\resultbox{To answer R0:
\ddmin is the most efficient domain-agnostic single core computation algorithm
in our setup.
}

\subsection{Results: Efficiency of \quickcore versus \ddmin}
\label{sec:evaluation:qctime}

\begin{table}\small
\caption{\small Efficiency of \quickcore vs. \ddmin}\vspace{-.5em}
\label{tbl:qc}
\resizebox{\columnwidth}{!}{
\begin{tabular}{l || r | r | r | r | r}
Spec set & $\leq 0.1s$ &  0.1-1s & 1-10s & 10-100s & $\geq 100s$ \\
\hline
SYN15U & 0.79 & 0.51 &-&-&-\\
SYN17U & 1.38  & 0.37 & 0.54 &0.06& 0.015 \\
\hline
AM+GN & 1.53 & 0.63 & 0.71 & 0.30 & 0.20 \\

\end{tabular}
}
\vspace{-1em}
\end{table}

%

Table~\ref{tbl:qc} presents the performance of \quickcore versus \ddmin.
We chose to compare \quickcore with \ddmin because \ddmin is a well known and
widely used algorithm for core computation over a monotonic criteria, and
because it was previously used in the context of unrealizable cores for GR(1)
specifications. Moreover, in R0 above we showed that \ddmin is the most efficient
domain-agnostic single core algorithm on our corpus.

The columns show the geometric mean of the ratio of the
running times (namely, the running time of \quickcore divided by that of
\ddmin), dissected by the running time range obtained for \ddmin, for all
specifications in each set. We use \texttt{`-'} to mark cases in which no
specifications had \ddmin running time within the corresponding range. For
example, the number 0.30 in the third row means that for \texttt{AM+GN}
specifications for which a core was found by \ddmin in between 10 and 100
seconds, the geometric mean indicates that \quickcore was more than three times
faster than \ddmin.


The results show that \quickcore is in most cases much faster than \ddmin.
This improvement gets better with scale, i.e., almost consistently, the slower
\ddmin, the faster \quickcore becomes relative to it.  The acceleration is most
noticeable for \texttt{SYN17U} specifications for which \ddmin require over 10
seconds. For those specifications \quickcore was faster than \ddmin by well over an order
of magnitude.

The only two categories in which the running time of \quickcore is worse than
that of \ddmin is for \texttt{SYN17U} and \texttt{AM+GN} specifications whose \ddmin 
running time is at most 100 milliseconds. 
Since for this range running times are very small, 
we do not consider it to be a major weakness of \quickcore.

We also computed the actual number of realizability checks (i.e., without
realizability checks avoided by memoization, see Sect.~\ref{sec:memo}) of
\quickcore and \ddmin (not shown in the table). We found that the median
reduction in the number of actual realizability checks of \quickcore over \ddmin
was 11.3\%, 19.5\%, and 15.4\%, and over \qxplain was 3.3\%, 10.2\%, and
27.6\%, for \texttt{SYN15U}, \texttt{SYN17U}, and \texttt{AM+GN}, respectively.

\vspace{.5em}
\resultbox{To answer R1:
\quickcore typically performs fewer realizability checks than \ddmin and
\qxplain, it is in most cases much faster than \ddmin, and the running time
improvement seems to become better with scale.
}

\subsection{Results: Effectiveness of \quickcore}
\label{sec:evaluation:qcsize}

\begin{table}\small\center
\caption{\small Core sizes}\vspace{-.5em}
\label{tbl:size}
\begin{tabular}{l || r || r | r | r | r }
Spec set & Core Size & \texttt{QC} & \ddmin & \texttt{QX} & \texttt{Global}\\
\hline
SYN15U &19\% & 5.14 & 5.07 & 5.21 & 5.07 \\
SYN17U &18\% & 4.5 & 4.38 & 4.5  & 3.92 \\
\end{tabular}
\vspace{-1em}
\end{table}

Table~\ref{tbl:size} presents core size results.
Column \texttt{Core size} shows the median ratio between the size of cores found
by \quickcore and the total number of guarantees in the specification.
Columns \texttt{QC}, \ddmin, \texttt{QX}, and \texttt{Global} show the average absolute size of
the cores found by \quickcore, \ddmin, \qxplain, and the size of the smallest
core found until the timeout by \punch, respectively.


The results show that in the \texttt{SYNTECH} specifications, most of the cores
found by \quickcore are over five times smaller than the total number of
guarantees in the specification. They further show that the cores found by
\quickcore have a slightly different size than the cores found by the other
algorithms, and that all algorithms output cores that are close in size to the
size of the globally minimal core.

\vspace{.5em}
\resultbox{To answer R2:
\quickcore, \ddmin, and \qxplain are 
all effective in localizing unrealizability.}

\subsection{Results: Number of Cores in Specifications}
\label{sec:evaluation:mult}

\begin{table}\small\center
\caption{\small Number of cores}\vspace{-.5em}
\label{tbl:howmany}
\begin{tabular}{l || r | r || r | r | r | r}
Spec set & S &  M & $\geq 5$ & $\geq 10$ & $\geq 50$ & $\geq 100$\\
\hline
SYN15U & 6  &  8 & 5 & 2 & 1 & 0\\
SYN17U & 6  & 20 & 9 & 7 & 4 & 4\\
\end{tabular}
\vspace{-1em}
\end{table}

Table~\ref{tbl:howmany} presents the number of cores in the
\texttt{SYNTECH} specifications. 
Columns \texttt{`S'} and \texttt{`M'} show how many specifications have a single core and multiple cores, resp. The
remaining four columns show how many of these specifications have at least 5,
10, 50, and 100 cores. For example, the number 1 in the first row under
\texttt{`$\geq 50$'} means that exactly one specification of \texttt{SYN15U} has
at least 50 cores.

\vspace{.5em}
\resultbox{To answer R3: 
Most \texttt{SYNTECH} specifications have multiple cores. 
Specifications with over 50 cores exist in each set of specifications.
These results motivate the need to compute more than one core per specification.}

\subsection{Results: Running Times of All Cores Algorithms} 

In order to evaluate the performance of \punch for computing all unrealizable
cores, we use \texttt{PQC} and \texttt{PUD} (see Sect.~\ref{subsec:instances}).



Since \punch is the first algorithm employed to compute all unrealizable
cores for temporal specifications reducible to GR(1), as a comparison, we use a
rather naive baseline we label \texttt{TD} (we discuss alternative algorithms in
Sect.~\ref{sec:related}).

\texttt{TD} is a naive top down search for all cores.
For a given subset of guarantees, if \texttt{check} fails, \texttt{TD} knows
that the subset and all its subsets are not cores.
Otherwise, it continues recursively to all subsets that exclude exactly one
element. It detects the set as a core iff all of these subsets fail
\texttt{check}. It memoizes subsets we already finished computing all cores for,
to avoid unnecessary recursive calls (in addition to the memoization described
in Sect.~\ref{sec:memo}). To detect unrealizable cores, \texttt{check} is
implemented as an unrealizability check.

\begin{figure}
  \centering
  \includegraphics[width=\linewidth]{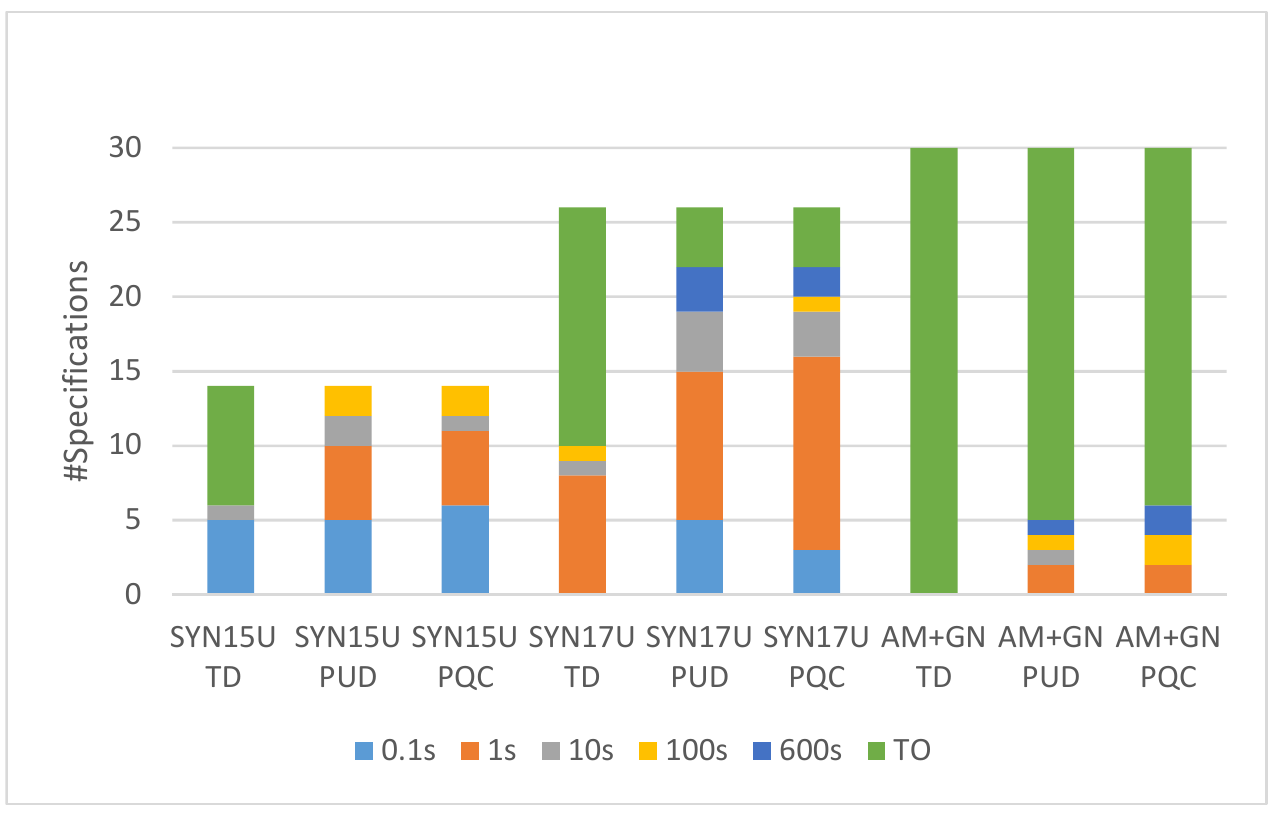}
  \vspace{-2em}
  \caption{\small Running times to compute all cores using \texttt{TD} (left columns) \texttt{PUD} (center columns), and \texttt{PQC} (right columns),
  for the SYNTECH and the \texttt{AM+GN} sets,
  divided by increasing ranges, in seconds.}
  \vspace{-1em}
  \label{fig:perform}
\end{figure}

Figure~\ref{fig:perform} shows the running times for \texttt{SYN15U},
\texttt{SYN17U}, and \texttt{AM+GN} specifications. Each specification set has
three columns, one for \texttt{TD}, one for \texttt{PUD}, and one for
\texttt{PQC}. Each column shows a breakdown of how many specifications completed
running within 0.1, 1, 10, 100, and 600 seconds.
Finally, we mark the number of specifications that timed out (did not complete
all 10 runs within 10 minutes) with \texttt{TO}.

The results show that the differences between the performance of 
\texttt{PUD} and \texttt{PQC} are minor, with mostly a slight advantage to \texttt{PQC}.
%
Regardless of their minor differences, the two instances of \punch perform
significantly better than \texttt{TD}. The two instances of \punch were able to
find all cores of \texttt{SYN15U} specifications within 100 seconds, and 22
out of the 26 of the \texttt{SYN17U} specifications before the timeout was
reached. This is much better than \texttt{TD}, which was able to find all cores within the timeout
for less than half of the specifications. On \texttt{AM+GN}
specifications, \texttt{PUD} and \texttt{PQC} found all cores for 5 and 6 specifications respectively, but
\texttt{TD} did not complete the computation on any of the specifications.

\vspace{.5em}
\resultbox{To answer R4: 
Both instances of \punch are significantly faster than \texttt{TD}
in finding all the cores of a specification.
\texttt{PQC} seems slightly better than \texttt{PUD}.}

\subsection{Additional Results}
\label{subsec:add}



Recall that \punch provides early detection of the intersection of all the cores
(see Sect.~\ref{subsec:punchcomplexity}).
The size of the intersection is a lower bound on the size of the global minimum.
In practice, it provides a good early estimate of the size of the global minimum
for \texttt{SYNTECH} specifications.
The results show
that the core intersection size is on average 76.8\% and 52.5\% of the smallest
core found by \punch within the timeout for \texttt{SYN15U} and \texttt{SYN17U}
respectively.


\subsection{Threats to Validity}

We discuss threats to the validity of our results. First, symbolic computations
are not trivial and our implementation of \quickcore, \ddmin, \qxplain, and
\punch may contain bugs. To mitigate, we performed a thorough validation using
all specifications available to us, see Sect.~\ref{subsec:validation}. Second,
we have based most of our evaluation on the \texttt{SYNTECH} specifications,
which were created by 3rd year undergraduate CS students with no prior
experience in writing LTL specifications (collected by the authors
of~\cite{MaozR21} in classes they have taught).
We further examined specifications from the \texttt{AM+GN} set. We do not know
if these are representative of specifications engineers would write in practice.
Third, although we have found that cores are typically much smaller than the
complete set of guarantees (see Tbl.~\ref{tbl:size}, roughly 5 guarantees
instead of 25), we did not perform a user-study, with engineers, to examine
whether users will find the reported cores useful for understanding the reasons
for the unrealizability of their specifications.

%
%
%
%

\section{Related Work}
\label{sec:related}

\subsection{A Single Unrealizable Core for GR(1)}

Previous works have considered the computation of an unrealizable core for GR(1)
specifications. Cimatti et al.~\cite{CimattiRST08} have used \texttt{LinearMin}
(see Sect.~\ref{subsec:ddmin}). Konighofer et al.~\cite{KonighoferHB13} have
used \ddmin and implemented it in the RATSY synthesizer. Their comparison of
\ddmin with \texttt{LinearMin} shows that \ddmin is almost always much faster
than \texttt{LinearMin}, with a greater advantage on larger specifications.
Firman et al.~\cite{FirmanMR20} have used \ddmin with several performance
heuristics, including memoization, and implemented it in the Spectra
synthesizer. All these were limited to computing a single core and did not
correctly handle specifications with constructs beyond pure GR(1). We present
\quickcore to be used instead of \ddmin. We further show how to correctly handle
specifications with constructs beyond pure GR(1). We compare \quickcore to
Spectra's implementation of \ddmin, i.e., with the heuristics
from~\cite{FirmanMR20}, and our evaluation provides evidence that \quickcore is
faster and scales better than \ddmin.

Our choice of \ddmin for both the algorithm we compare to, and the algorithm
used within \quickcore for the incremental minimization, is based both on \ddmin being
a well known and widely used domain-agnostic minimization algorithm, and on the fact
that it was the choice in previous work on GR(1) unrealizability. 

There are other domain-agnostic minimization algorithms over monotonic criteria
for single cores, e.g., \qxplain~\cite{Junker04, MarquesJB13}.  Compared to
\ddmin, \qxplain has a better asymptotic complexity in terms of the number of
checks (see Sect.~\ref{subsec:ddmin}). Nevertheless, as we show in our evaluation
(Sect.~\ref{subsec:qxvsddmin}), \qxplain performs worse on our corpus.


For temporal specifications, Schuppan~\cite{Schuppan12} presented LTL unsatisfiability cores by weakening
LTL formulas in a way that ignores sub-formulas not required for unsatisfiability.
He further presented a similar approach for GR(1) unrealizability cores.
To the best of our knowledge, these ideas have only been explored theoretically.
Moreover, the work neither handles all cores nor deals with extensions of GR(1).

\subsection{All Unrealizable Cores for GR(1)}

We present \punch as the first efficient algorithm to compute all
unrealizable cores of specifications reducible to GR(1). However, \punch is a
domain-agnostic algorithm. Other domain-agnostic algorithms for all cores
computations over monotonic criteria appear in the literature, see,
e.g.,~\cite{LiffitonPMM16}. In~\cite{BendikC18} there is a comparison of several
such algorithms, which concludes that none of the known algorithms is better
than the others in all domains. Recently, MUST~\cite{BendikC20} was proposed as
an algorithm and tool that outperforms previous ones.

\punch is intended as a first algorithm for the computation of all cores of
unrealizable GR(1) specifications. We consider its comparison against
domain-agnostic algorithms, as well as its specialization for GR(1) as future
work, see Sect.~\ref{sec:conclusion}.

\subsection{Other Approaches to Dealing with GR(1) Unrealizability}

Beyond core computations, other approaches have been suggested to dealing with
unrealizability of GR(1) specifications. Maoz and Sa'ar~\cite{MaozS13} have
presented the computation of counter-strategies, which show how the environment
can prevent any system from satisfying the specification. Kuvent et
al.~\cite{KuventMR17} have presented the JVTS, a symbolic, more succinct and
simple representation of a GR(1) counter-strategy. 


Other works have considered means to repair unrealizable specifications by
automatically suggesting additional assumptions that will make the specification
realizable, see,
e.g.,~\cite{AlurMT13,CavezzaA17,CavezzaAG18,CavezzaAG20,MaozRS19}. It may be
possible to combine the computation of a core or of all cores with a repair
approach. See our discussion of future work in Sect.~\ref{sec:conclusion}.

\section{Conclusion and Future Work}
\label{sec:conclusion}

We presented three contributions related to the computation of unrealizable
cores of GR(1) specifications, including faster algorithms for computing an
unrealizable core and for computing all cores. We further presented means to
correctly compute the core when specifications include high-level
constructs.

We implemented our work, validated its correctness, and evaluated it on
benchmarks from the literature. The evaluation shows that \quickcore is usually
faster than previous algorithms, with a relative advantage that improves with
scale. Moreover, we found that most specifications have multiple cores, and that
\punch finds all the cores significantly faster than a competing naive algorithm.

Our work has important implications to anyone using GR(1) specifications and
their extensions for synthesis and related analyses. First, core computations
are now faster, and computing more than one core promotes a more comprehensive
view of unrealizability. Moreover, we handle higher-level constructs correctly
and efficiently, and our algorithm for finding all cores extends to cores for
any monotonic criterion.

We consider the following concrete future work directions.
First, as \punch in its raw form is domain-agnostic, it is important to compare its
performance with recent domain-agnostic all cores minimization algorithms over
monotonic criteria such as MUST~\cite{BendikC20}.

We have already presented a variant of \punch, namely \texttt{PQC}, which
employs a domain-specific algorithm (a variant of \quickcore) to compute a
single core, and compared it with another variant, namely \texttt{PUD}, which
uses \ddmin, a domain-agnostic algorithm for a single core computation.
The results showed that \texttt{PQC} performs slightly better than \texttt{PUD}.
It may be possible to improve the performance of detecting all unrealizable
cores by taking advantage of GR(1) specific properties in other ways.

Second, we consider means to combine the computation of a core or of all cores
with a repair approach. For example, a repair of unrealizability that is based on weakening a small
as possible subset of the guarantees could rely on the fact that it cannot
succeed without weakening at least one guarantee from every core, or in particular
weakening one guarantee from a non-empty intersection of all cores.

Finally, the ability to compute all cores raises questions as to their
presentation to the engineer. Should all cores be computed and presented?
Perhaps an on demand approach should be used? In which order should we present
the cores? These questions call for further investigation and evaluation.



    
    

\section*{Acknowledgements}
We thank the anonymous reviewers for their helpful comments.
We thank Roee Sinai for implementing the \qxplain algorithm in the Spectra environment.
We thank Inbar Shulman for comments about an earlier draft of the paper.
This project has received funding from the European Research Council (ERC)
under the European Union's Horizon 2020 research and innovation programme (grant
agreement No 638049, SYNTECH).

\newpage
\bibliographystyle{abbrv}
\bibliography{doc}

\end{document}